\documentclass{article}
\usepackage{ijcai20}

\usepackage{times}

\usepackage{soul}
\usepackage{url}
\usepackage[hidelinks]{hyperref}
\usepackage[utf8]{inputenc}
\usepackage[small]{caption}
\usepackage{graphicx}
\usepackage{amsmath}
\usepackage{booktabs}
\usepackage{ntheorem}
\usepackage{makecell}
\newtheorem{prop}{Proposition}
\newtheorem{defn}{Definition}
\newtheorem{lemma}{Lemma}
\newtheorem{theorem}{Theorem}
\newtheorem*{proof}{Proof}
\newtheorem*{rem}{Remark}
\usepackage{amsfonts}
\usepackage[vlined,boxed,ruled]{algorithm2e}
\usepackage{enumerate}
\usepackage{units}

\title{Opinion Maximization in Social Trust Networks}

\author{
Pinghua Xu$^{1,2}$\and
Wenbin Hu$^{1,2}$\footnote{Corresponding authors.}\and
Jia Wu$^3$\and
Weiwei Liu$^1$\footnotemark[1]\\
\affiliations
$^1$School of Computer Science, National Engineering Research Center for Multimedia Software and Institute of Artificial Intelligence, Wuhan
University\\
$^2$Shenzhen Research Institute, Wuhan University\\
$^3$Department of Computing, Macquarie University\\
\emails
\{xupinghua, hwb\}@whu.edu.cn,
jia.wu@mq.edu.au,
liuweiwei863@gmail.com
}

\begin{document}

\maketitle

\begin{abstract}
Social media sites are now becoming very important platforms for product promotion or marketing campaigns. Therefore, there is broad interest in determining ways to guide a site to react more positively to a product with a limited budget. However, the practical significance of the existing studies on this subject is limited for two reasons. First, most studies have investigated the issue in oversimplified networks in which several important network characteristics are ignored. Second, the opinions of individuals are modeled as bipartite states (e.g., support or not) in numerous studies, however, this setting is too strict for many real scenarios. In this study, we focus on social trust networks (STNs), which have the significant characteristics ignored in the previous studies. We generalized a famed continuous-valued opinion dynamics model for STNs, which is more consistent with real scenarios. We subsequently formalized two novel problems for solving the issue in STNs. Moreover, we developed two matrix-based methods for these two problems and experiments on real-world datasets to demonstrate the practical utility of our methods.
\end{abstract}

\section{Introduction}

On social sites, users can express their opinions on a product, a social event, or many other information items.
By quantifying opinions with numerical values, the overall opinion of a network can be determined as the sum total of the opinions of all individuals. This value reflects the emotional tendency of the site toward an information item. Knowing the overall opinion is critical for many economic and academic applications, such as marketing campaigns and public voice control. 
For example, suppose that we are running a mobile phone company and have just released a new phone. In order to boost sales, it is vital to guide the market to react more positively (i.e., have a more positive overall opinion) to our product with a limited budget. We call this issue \textit{opinion maximization}. It is the general form of influence maximization \cite{li2018influence} and the main difference is that we consider opinions in continuous-valued states rather than bipartite states, which is too strict in many real scenarios.

\begin{figure}[t]
	\centering
	\includegraphics[scale=0.45]{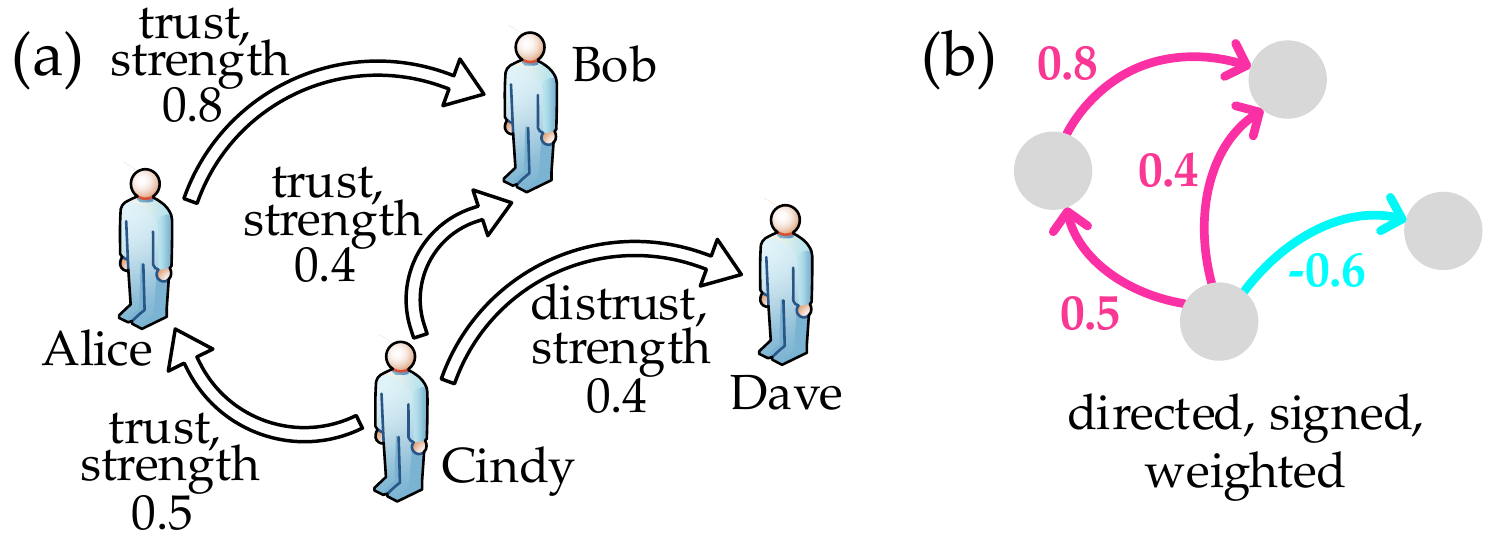}
	\caption{(a) A social trust network. (b) Graph representation of (a).}
	\label{f:STN}
\end{figure}

Different from most existing studies on opinion or influence maximization \cite{chen2016robust,yang2016continuous,Abebe2018opinion,liu2018active}, we
argue that it is more realistic and valuable to study this issue in a social trust network (STN, see Figure \ref{f:STN}), which models the individuals and the trust or distrust relationships between them and can be represented by a \textit{directed} \textit{signed} \textit{weighted} graph \cite{kumar2016edge,xu2019social}. The main reason for this is the fact that the following three important network characteristics have usually not been reflected in the target networks of the existing studies simultaneously.
\begin{itemize}
    \item A lot of normal users tend to be influenced by key opinion leaders (KOLs), but KOLs are hardly influenced by normal users. This characteristic is reflected in the directionality of STNs.
    
    \item Users can be influenced by others to varying degrees and this characteristic is reflected in the absolute weights of edges in STNs.
    
    \item Relationships between users may have opposite polarities. This characteristic is reflected in the signs of the edges' weights in STNs.
\end{itemize}

Given the target networks, we need to model how opinions of individuals form, then we can design strategies to maximize the overall opinion. Therefore, we have chosen to generalize a famed \textit{continuous-valued} opinion dynamics model \cite{Friedkin1990Social} for STNs for three reasons:
\begin{itemize}
    \item Opinions cannot be accurately modeled by discrete states, such as being absolutely supportive or opposed, in many real scenarios. For example, one may give a little support or strong support to a new iPhone.
    
    \item Opinion dynamics better resembles a social game in which individuals constantly influence each other rather than a discrete cascading process \cite{Gionis2013Opinion}.
    
    \item In addition to the trust relationships, distrust relationships also play an important role in opinion dynamics \cite{li2013influence,li2017positive}. But they are usually ignored in existing studies of opinion dynamics.
\end{itemize}

As an STN has several complex network characteristics and the opinion dynamics model used is more realistic and complex, opinion maximization becomes much more difficult in STNs. Therefore, existing studies are not applicable for STNs.
In this study, we formalize two novel problems for solving opinion maximization in STNs. They are, with a limited budget, how to maximize the overall opinion by changing internal opinions of some people (IOP) or by fixing the expressed opinions of some people (EOP). In IOP, we observe that people play different roles in the social game, i.e., their internal opinions contribute differently to the overall opinion. This phenomenon is not observed in oversimplified networks. In addition, we define a novel contribution index to quantify the contribution ability, then design a simple but effective method for IOP. The expressed opinion problem is NP-hard, hence we design a greedy method to achieve an acceptable solution. Iteratively calculating the benefit of fixing a node's expressed opinion is time-consuming, but we integrate the calculations into a simple expression via crafty matrix operations. That makes our method efficient. We evaluate our methods on several real-world datasets and the experimental results demonstrate the practical utility of the proposed methods. The contributions of this study can be summarized as follows.\footnote{The implementations of the methods are available at https://github.com/WHU-SNA/OpMaxInSTN}
\begin{itemize}
    \item This is the first study that investigates opinion maximization in STNs with a continuous-valued opinion dynamics model. We formalize two novel problems, IOP and EOP, for maximizing the overall opinion of an STN.
    
    \item We define a novel contribution index to quantify the nodes' variant contribution abilities to the overall opinion, which are unobserved in oversimplified networks. We also develop a simple method for IOP that can easily achieve optimal performance. 
    
    \item We prove the hardness of EOP. By integrating the vast computation via crafty matrix operations, we develop an efficient method for EOP.
\end{itemize}

The rest of this paper is organized as follows. The preliminaries include an explanation of our notation and the opinion dynamics model is described in Section 2. The problems and the proposed methods are presented in Section 3. Section 4 details the experiments, followed by the conclusions in Section 5.

\section{Preliminaries}

\subsection{Notation}
A social trust network is represented by a \textit{directed signed weighted} graph $\mathcal{G}=(\mathcal{V},\mathcal{E})$, where $\mathcal{V}$ is the set of nodes and $\mathcal{E}$ is the set of edges. Each directed edge from $i$ to $j$, $(i,j) \in \mathcal{E}$, is associated with the weight $w_{ij}\in [-1,1]$, which signifies the strength of the relationship. Let $\textbf{A}$ denote the adjacency matrix of $\mathcal{G}$, where $a_{ij}=w_{ij}$ if $(i,j)\in \mathcal{E}$, $a_{ij}=0$ otherwise. Let $\textbf{D}$ denote the row connectivity matrix, where $d_{ii}=\begin{matrix} \sum_{j} |a_{ij}| \end{matrix}$. Then $\bar{\textbf{L}}=\textbf{D}-\textbf{A}$ denotes the generalized Laplacian matrix \cite{hu2013bipartite}.

\begin{lemma}
Inequality $0 \leq \min sp(\bar{\textbf{L}}_u) \leq Re(sp(\bar{\textbf{L}})) \leq \max sp(\bar{\textbf{L}}_u) $, where $\bar{\textbf{L}}_u = \nicefrac{(\bar{\textbf{L}} + \bar{\textbf{L}}^T)}{2}$, $sp(*)$ indicates the spectra and $Re(*)$ indicates the real parts.
\label{lemma:1}
\end{lemma}

\begin{theorem}
The term $(\bar{\textbf{L}}+\textbf{I})$ is invertible, where $\textbf{I}$ is identity matrix.
\label{theorem:foundamentalMatrix}
\end{theorem}

\begin{proof}
Denote the eigenvalues of $\bar{\textbf{L}}$ as $\{k_1, k_2, ..., k_{|\mathcal{V}|}\}$ and the corresponding eigenvectors as $\{\textbf{x}_1, \textbf{x}_2, ..., \textbf{x}_{|\mathcal{V}|}\}$. Then, we have the following derivation:
\begin{equation*}
    (\bar{\textbf{L}}+\textbf{I}) \textbf{x}_i = \bar{\textbf{L}}\textbf{x}_i + \textbf{x}_i = (k_i+1)\textbf{x}_i
\end{equation*}

Hence, the eigenvalues of $(\bar{\textbf{L}}+\textbf{I})$ are $\{k_1+1, k_2+1, ..., k_{|\mathcal{V}|}+1\}$. According to lemma \ref{lemma:1}, for $1 \leq i \leq |\mathcal{V}|$, $Re(k_i) \ge 0$. Therefore, the real parts of all the eigenvalues of $(\bar{\textbf{L}}+\textbf{I})$ are greater than $1$, $det(\bar{\textbf{L}}+\textbf{I}) \neq 0$, and $(\bar{\textbf{L}}+\textbf{I})$ is invertible.
\end{proof}

\subsection{Friedkin-Johnsen Model}
The Friedkin–Johnsen model \cite{Friedkin1990Social} is a famed continuous-valued opinion dynamics model. In this model, each node $i\in \mathcal{V}$  holds a fixed internal opinion $s_i$, which remains constant. During each iteration, node $i$ updates its expressed opinion $z_i$ as follows:
\begin{equation}
	\label{e:F-J model}
	z_i=\tfrac{s_i+\begin{matrix} \sum_{j\in \mathcal{N}(i)} w_{ij}z_j\end{matrix}}{1+\begin{matrix}\sum_{j\in \mathcal{N}(i)} w_{ij}\end{matrix}}
\end{equation}
where $\mathcal{N}(i)$ denotes the successors of i. The values of $s_i$ and $z_i$ are in the interval $[-1,1]$, where $-1$/$1$ signifies the strongest positive/negative opinion. The repeated averaging does converge to the Nash equilibrium of a social game \cite{bindel2015bad}.

\subsection{Generalized Opinion Dynamics Model For Social Trust Networks}

As relationships with opposite polarities in STNs are now considered, we generalize the social game of the Friedkin–Johnsen model according to structural balance theory \cite{islam2018signet}. Specifically, an individual is supposed to express the opposite opinions of the people she distrusts. Hence the consensus cost function of the social game is:
\begin{equation*}
    c(z_i)=(z_i-s_i)^2+\begin{matrix} \sum_{j\in \mathcal{N}(i)} |w_{ij}|(z_i-sgn(i,j)z_j)^2 \end{matrix}
\end{equation*}
where $sgn(i,j)$ denotes the sign of nonzero $w_{ij}$. Accordingly, the repeated averaging process becomes:
\begin{equation}
\label{e:signed F-J model}
z_i=\tfrac{s_i+\begin{matrix} \sum_{j\in \mathcal{N}(i)} w_{ij}z_j\end{matrix}}{1+\begin{matrix}\sum_{j\in \mathcal{N}(i)} |w_{ij}|\end{matrix}}
\end{equation}

\begin{prop}
Let $\textbf{s}$ denote the internal opinion vector, whose $i$-th component is $s_i$, and $\textbf{z}^*$ denote the expressed opinion vector at the Nash equilibrium. Equation (\ref{e:signed F-J model}) is solved by the equilibrium $\textbf{z}^*=(\bar{\textbf{L}}+\textbf{I})^{-1}\textbf{s}$.

\label{prop:expressedOpinion}
\end{prop}

\begin{proof}
At the Nash equilibrium, each node $i$ adopts an expressed opinion $z^*_i$ that minimizes its cost, which implies that $c'(z^*_i)=0$.
\begin{equation}
\label{e:new cost}
\begin{matrix}
	c'(z^*_i)=2(z^*_i-s_i)+2\begin{matrix} \sum_{j\in \mathcal{N}(i)} |w_{ij}|(z_i^*-sgn(i,j)z^*_j)\end{matrix}=0 \\
	\Rightarrow z_i^*=\tfrac{s_i+\begin{matrix} \sum_{j\in \mathcal{N}(i)} w_{ij}z_j^*\end{matrix}}{1+\begin{matrix}\sum_{j\in \mathcal{N}(i)} |w_{ij}|\end{matrix}}
\end{matrix}
\end{equation}

Therefore, the Nash equilibrium opinion of a node is the weighted average of its internal opinion, the Nash equilibrium opinions of the nodes it trusts, and the opposite Nash equilibrium opinions of the nodes it distrusts. We rewrite equation (\ref{e:new cost}) as:
\begin{equation*}
    (1+\begin{matrix}\sum_{j\in \mathcal{N}(i)} |w_{ij}|\end{matrix})z^*_i=s_i+\begin{matrix} \sum_{j\in \mathcal{N}(i)} w_{ij}z^*_j\end{matrix}
\end{equation*}

By introducing the opinion vectors $\textbf{s}$ and $\textbf{z}^*$, we can formalize equation (\ref{e:new cost})'s matrix form as:
\begin{equation*}
\begin{aligned}
(\textbf{D}+\textbf{I})\textbf{z}^*=\textbf{s}+\textbf{A}\textbf{z}^{*}
&\Rightarrow (\textbf{D}-\textbf{A}+\textbf{I})\textbf{z}^*=\textbf{s} \\
&\Rightarrow (\bar{\textbf{L}}+\textbf{I})\textbf{z}^*=\textbf{s}
\end{aligned}
\end{equation*}

According to Theorem \ref{theorem:foundamentalMatrix}, $(\bar{\textbf{L}}+\textbf{I})$ is invertible. Therefore, we obtain that $\textbf{z}^*=(\bar{\textbf{L}}+\textbf{I})^{-1}\textbf{s}$.

\end{proof}

Given the generalized opinion dynamics model for STNs, we adopt a similar definition of overall opinion introduced in \cite{Gionis2013Opinion}.

\begin{defn}
(Overall opinion). The overall opinion $p(\textbf{z}^*)$ of an STN is the sum of the expressed opinions at Nash equilibrium:
\begin{equation*}
	p(\textbf{z}^*)=\begin{matrix}\sum_i z^*_i\end{matrix}=\vec{\textbf{1}}^T(\bar{\textbf{L}}+\textbf{I})^{-1}\textbf{s}
\end{equation*}
where $\vec{\textbf{1}}$ is a column vector of ones.
\end{defn}

\section{Opinion Maximization in Social Trust Networks}

In this section, we formalize IOP and EOP for solving opinion maximization in STNs, and we develop two methods for IOP and EOP, respectively. Figure \ref{f:IOP_EOP_example} illustrates that we can increase the overall opinion by changing a node's internal opinion or by fixing a node's expressed opinion.

\begin{figure}[t]
	\centering
	\includegraphics[scale=0.25]{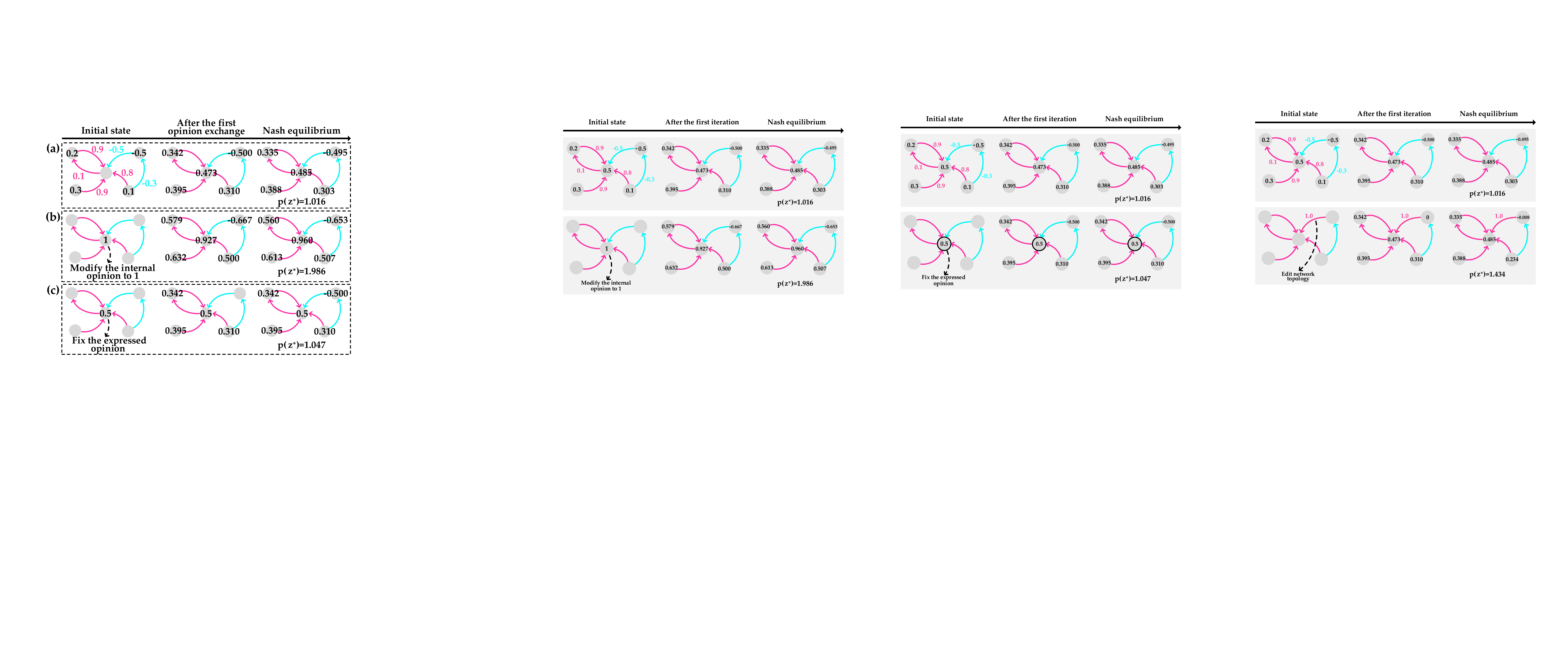}
	\caption{Subfigure (a) illustrates the opinion dynamics in an STN. The values on the nodes indicate the expressed opinions, and the values close to the edges indicate the weights. Note that the expressed opinion of a node equals its internal opinion at the initial state. We can increase the overall opinion by (b) changing the internal opinion of a node or (c) by fixing the expressed opinion of a node.}
	\label{f:IOP_EOP_example}
\end{figure}

\begin{figure}[t]
	\centering
	\includegraphics[scale=0.3]{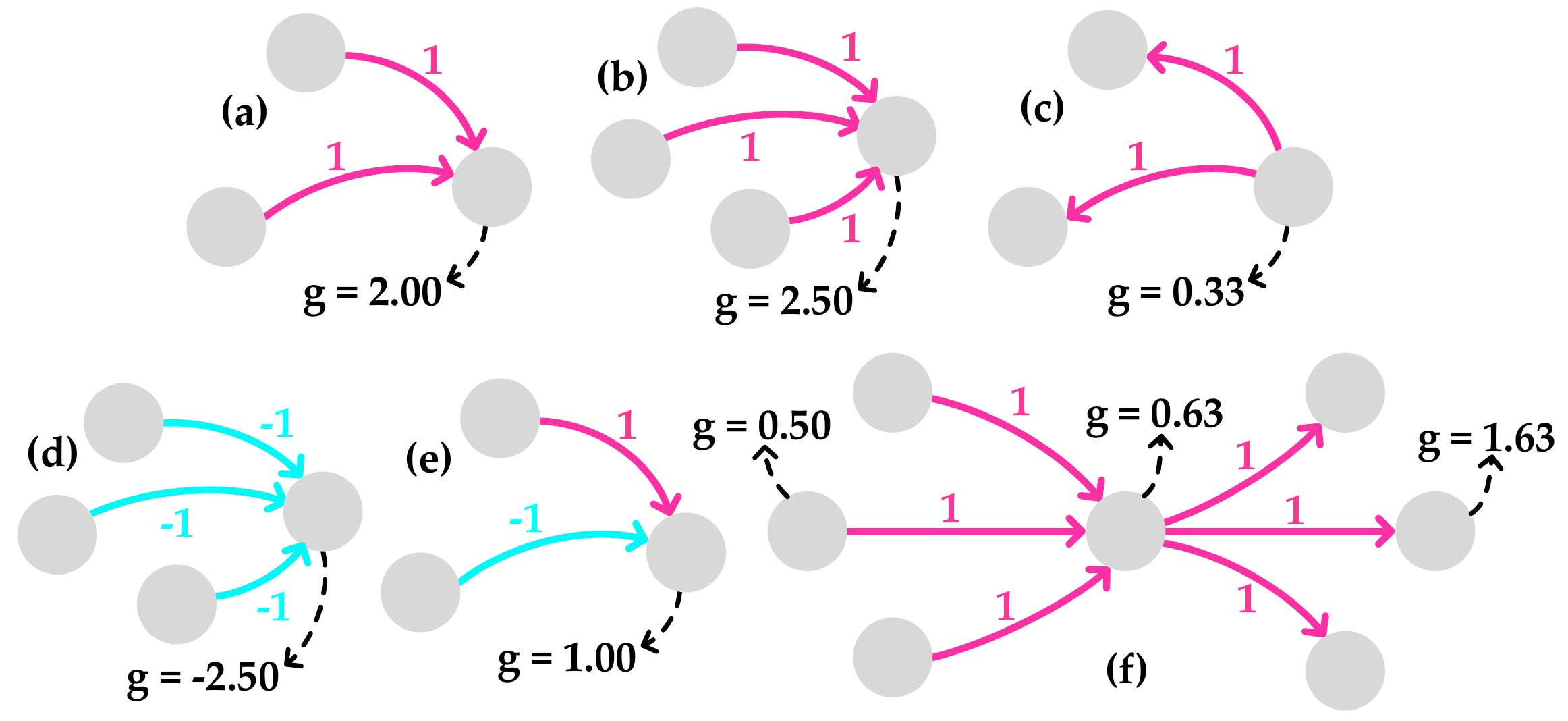}
	\caption{Example of the contribution index.}
	\label{f:Contribution_example}
\end{figure}

\subsection{Internal Opinion Problem}
In IOP, we seek to change the internal opinions of some nodes to specific values, such that the overall opinion is maximized. Our decision is limited by a budget $\mu$ on the total modification of internal opinions. Let $\Delta\textbf{s}$ denote the modification of the internal opinions and the formal problem definition is the following.

\begin{defn}
(Internal Opinion Problem). Given a social trust network $\mathcal{G}$, the internal opinion vector $\textbf{s}$, and the budget amount $\mu$, seek out the modification of internal opinions $\Delta\textbf{s}$, such that the overall opinion $p(\textbf{z}^*)$ is maximized:
\begin{align*}
&\max\limits_{\Delta \textbf{s}}\vec{\textbf{1}}^T(\bar{\textbf{L}}+\textbf{I})^{-1}(\textbf{s}+\Delta\textbf{s})\\
&s.t.\quad
\begin{cases}
-1\leq s_i+\Delta s_i\leq 1&(i=1,2,3,\cdots,|\mathcal{V}|) \\
\lVert \Delta\textbf{s} \rVert_1 \leq \mu
\end{cases}
\end{align*}
\label{defn:IOP}
\end{defn}

According to Proposition \ref{prop:expressedOpinion}, the expressed opinion vector is in the column space of $(\bar{\textbf{L}}+\textbf{I})^{-1}$ and the coordinates are stored in the vector of internal opinions. Then, we refer to $\textbf{g}=(\vec{\textbf{1}}^T(\bar{\textbf{L}}+\textbf{I})^{-1})^T$ as the contribution index vector. As $p(\textbf{z}^*)=\vec{\textbf{1}}^T(\bar{\textbf{L}}+\textbf{I})^{-1}\textbf{s}=\textbf{g}^T\textbf{s}$, each entry $g_i$ indicates the amount by which the node $i$'s internal opinion contributes to the overall opinion. Formally, the contribution index is defined as follows.

\begin{defn}
(Contribution index). The contribution index $g_i$ of node $i$ indicates how $i$'s internal opinion contributes to the overall opinion and is quantified by:
\begin{equation*}
	g_i=(\vec{\textbf{1}}^T(\bar{\textbf{L}}+\textbf{I})^{-1})^T_i
\end{equation*}
\end{defn}

In Figure \ref{f:Contribution_example}, we demonstrate some examples of the contribution index. For simplicity, the edge weight is set to $+1$ or $-1$. From (a)(b), we observe that a node, which receives more trusts from other nodes, has a larger contribution index. Subfigures (a)(c) illustrate that a KOL has more impact on the overall opinion than normal users. The target node in (d) has a negative contribution index, as it promotes the nodes who distrust it to adapt the opposite of its expressed opinion. In subfigure (e), the target node is trusted by a node and distrusted by another. The combined effects of the relationships with opposite signs result in a contribution index of 1. Subfigure (f) shows the contribution indices of nodes in a hierarchical network structure.

By changing the internal opinion of the node with the largest absolute value of the contribution index, we can get the most benefit at the same budget. Our method SIOP is summarized in Algorithm 1. Further, we proof that Algorithm 1 outputs the optimal solution of IOP.

\begin{theorem}
    Algorithm 1 outputs the optimal solution of IOP.
\end{theorem}

\begin{proof}
    Given the output of Algorithm 1, the benefit is $\textbf{g}^T\Delta\textbf{s}$. For any intervened node $i$, we can adjust the corresponding modification $\Delta s_i$ to $(\Delta s_i - \nu)$ if $\Delta s_i > 0$, or $(\Delta s_i + \nu)$ otherwise, where $\nu>0$ is very slight perturbation. We then have unused budget $\nu$. By changing the internal opinion of node $j$, which is not intervened before, the benefit becomes $\textbf{g}^T\Delta\textbf{s} + (|g_j| - |g_i|)\nu$. According to the procedure of Algorithm 1, $|g_j| \leq |g_i|$, thus $\textbf{g}^T\Delta\textbf{s} \ge \textbf{g}^T\Delta\textbf{s} + (|g_j| - |g_i|)\nu$.
\end{proof}

\begin{rem}
IOP may have multiple optimal solutions.
\end{rem}

\begin{algorithm}[t]
	\caption{Solving internal opinion problem}
	\KwIn{A social trust network $\mathcal{G}=(\mathcal{V},\mathcal{E})$, internal opinions $\textbf{s}$, and budget $\mu$.}
	\KwOut{The modification $\Delta\textbf{s}$.}
	
	Initialize $\Delta\textbf{s}$ to null vector;
	
	\While{$\mu >0$}
	{
		Find the component $g_i$ with the largest absolute value;
		
		\If{$g_i=0$}{break;}
		
		$sign = g_i/|g_i|$;
		$cost = 1 - sign \cdot s_i$;
		
		\eIf{$cost>\mu$}{
			$\Delta s_{i}=sign \cdot \mu$;
			break;
		}{
			$\Delta s_{i}=sign \cdot cost$;
			$\mu =\mu-cost$;	
		}
		
	}	
\end{algorithm}

\subsection{Expressed Opinion Problem}
Suppose that we can fix the expressed opinions of some nodes to their internal opinions during the opinion formation process. The chosen nodes will then never be affected by the expressed opinions of their neighbors. In EOP, we try to find a set of nodes $\mathcal{U}\in \mathcal{V}$. For each node $i$ in $\mathcal{U}$, we fix its expressed opinion, such that the overall opinion is maximized. Let $p(\textbf{z}^*|\mathcal{U})$ denote the overall opinion after fixing the expressed opinion of nodes in $\mathcal{U}$ and the formal problem definition is the following.

\begin{defn}
(Expressed opinion problem). Given a social trust network $\mathcal{G}$, the internal opinion vector $\textbf{s}$, and an integer $\mu$, seek a set $\mathcal{U}$ of $\mu$ nodes and fix $z_i$ to $s_i$ for $i\in \mathcal{U}$, such that $p(\textbf{z}^*|\mathcal{U})$ is then maximized.
\end{defn}

\begin{theorem}
EOP is NP-hard.
\label{theorem:hardness}
\end{theorem}

\begin{proof}
Due to lack of space, we only sketch the proof here. Our proof follows the idea in \cite{Gionis2013Opinion}. We generalized the absorbing random walk for STNs, then constructed a reduction from the problem of vertex cover on regular graphs.
\end{proof}

Given the hardness of EOP, we designed a greedy algorithm to achieve an acceptable solution. The algorithm starts with an empty set $\mathcal{U}^0$, and extends $\mathcal{U}^{(t-1)}$ by adding a node $i\in \mathcal{V}\setminus \mathcal{U}^{(t-1)}$ in the $t$-th iteration, such that the benefit to the overall opinion $p(\textbf{z}^*|\mathcal{U}^{(t-1)}\cup {i})-p(\textbf{z}^*|\mathcal{U}^{(t-1)})$ is maximized.
In the first iteration, we successively calculate the benefit of fixing the expressed opinion of each node $i$ in $\mathcal{V}$. Note that from the view of linear algebra, fixing $i$'s expressed opinion is equivalent to replacing the $i$-th row of $\textbf{A}$ with a row vector of zeros. And the row connectivity changes accordingly.

Let $\textbf{Q}=(\bar{\textbf{L}}+\textbf{I})^{-1}$ denote the fundamental matrix, $\textbf{Q}^\mathcal{U}$ denote the fundamental matrix after fixing the expressed opinions of nodes in $\mathcal{U}$. Obviously, $\textbf{Q}^{\mathcal{U}^0}=\textbf{Q}$. According to Theorem \ref{theorem:EOP_invertible}, we compute $\textbf{Q}^{\mathcal{U}^0\cup {i}}$ (equals to $\textbf{Q}^{\{i\}}$) using the Sherman–Morrison formula:
\begin{equation*}
\textbf{Q}^{\{i\}}=((\bar{\textbf{L}}+\textbf{I})+(-\textbf{e}_i \textbf{l}_{i:}))^{-1}\\
=\textbf{Q}+\frac{\textbf{Q}\textbf{e}_i\textbf{l}_{i:}\textbf{Q}}{1-\textbf{l}_{i:} \textbf{Q}\textbf{e}_i}
\end{equation*}
where $\textbf{e}_i$ is the unit vector, whose $i$-th component is 1, and $\textbf{l}_{i:}$ is the $i$-th row of $\bar{\textbf{L}}$. In the $t$-th iteration, the fundamental matrix after fixing the expressed opinion of $i\in \mathcal{V}\setminus \mathcal{U}^{(t-1)}$is the following:
\begin{equation}
	\label{e:updateFoundmental}
	\textbf{Q}^{\mathcal{U}^{t-1}\cup i}=\textbf{Q}^{\mathcal{U}^{t-1}}+\frac{\textbf{Q}^{\mathcal{U}^{t-1}}\textbf{e}_i\textbf{l}_{i:}\textbf{Q}^{\mathcal{U}^{t-1}}}{1-\textbf{l}_{i:}\textbf{Q}^{\mathcal{U}^{t-1}}\textbf{e}_i}
\end{equation}

\begin{theorem}
The term $\textbf{X}=(\bar{\textbf{L}}+\textbf{I})+(-\textbf{e}_i \textbf{l}_{i:})$ is invertible.
\label{theorem:EOP_invertible}
\end{theorem}

\begin{proof}
Expand $\textbf{X}$ along the $i$-th row and we find:
\begin{equation*}
	det(\textbf{X})=\begin{matrix} \sum_{j=1}^{|\mathcal{V}|} (-1)^{i+j}x_{ij}M_{ij} \end{matrix}=(-1)^{i+i}M_{ii}=M_{ii}
\end{equation*}
where $M_{ij}$ is the minor of $\textbf{X}$. Denote the Laplacian matrix of sub-graph without node $i$ as $\bar{\textbf{L}}'$, then we have:
\begin{equation*}
	M_{ii} = det(\bar{\textbf{L}}' + \textbf{I} + \textbf{C})
\end{equation*}
where $\textbf{C} \in \mathbb{R}^{(|\mathcal{V}|-1) \times (|\mathcal{V}|-1)}$ is a diagonal matrix, and $c_{jj} = |a_{ji}|$ for $j<i$ and $c_{(j-1)(j-1)} = |a_{(j-1)i}|$ for $j>i$.
Like the proof of Theorem \ref{theorem:foundamentalMatrix}, we find $Re(sp(\bar{\textbf{L}}' + \textbf{I} + \textbf{C})) \geq 1$. Therefore, $det(\bar{\textbf{L}}' + \textbf{I} + \textbf{C}) \neq 0$. As $det(\textbf{X})=M_{ii}=det(\bar{\textbf{L}}' + \textbf{I} + \textbf{C})$, the term $\textbf{X}$ is invertible.
\end{proof}

Thus, the benefit of fixing the expressed opinion of $i$ in the $t$-th iteration is:
\begin{equation}
\label{e:eop}
\begin{aligned}
\vec{\textbf{1}}^T(\textbf{Q}^{\mathcal{U}^{t-1}\cup i}-\textbf{Q}^{\mathcal{U}^{t-1}})\textbf{s}&=\vec{\textbf{1}}^T\frac{\textbf{Q}^{\mathcal{U}^{t-1}} \textbf{e}_i \textbf{l}_{i:} \textbf{Q}^{\mathcal{U}^{t-1}}}{1-\textbf{l}_{i:} \textbf{Q}^{\mathcal{U}^{t-1}} \textbf{e}_i}\textbf{s} \\
\end{aligned}
\end{equation}

We first calculate the term $\textbf{e}_i \textbf{l}_{i:} \textbf{Q}^{\mathcal{U}^{t-1}}$ through:
\begin{equation*}
\begin{aligned}
&\textbf{e}_i \textbf{l}_{i:} \textbf{Q}^{\mathcal{U}^{t-1}} = \textbf{e}_i ((\textbf{l}_{i:} + \textbf{e}_i^T) - \textbf{e}_i^T) \textbf{Q}^{\mathcal{U}^{t-1}} \\
=&\textbf{e}_i(\textbf{Q}^{\mathcal{U}^{t-1}})^{-1}_{i:}\textbf{Q}^{\mathcal{U}^{t-1}}-\textbf{e}_i \textbf{e}_i^T \textbf{Q}^{\mathcal{U}^{t-1}}=\textbf{e}_i (\textbf{e}_i^T - \textbf{Q}^{\mathcal{U}^{t-1}}_{i:})
\end{aligned}
\end{equation*}

This result illustrates that $\textbf{e}_i \textbf{l}_{i:} \textbf{Q}^{\mathcal{U}^{t-1}}$ is the matrix whose only non-negative row is the $i$-th row. And its $i$-th row is the negative $\textbf{Q}^{\mathcal{U}^{t-1}}_{i:}$ with the addition of one in the $i$-th entry.

Then, we calculate the term $\textbf{l}_{i:} \textbf{Q}^{\mathcal{U}^{t-1}} \textbf{e}_i$ through:
\begin{equation*}
\begin{aligned}
& \textbf{l}_{i:} \textbf{Q}^{\mathcal{U}^{t-1}} \textbf{e}_i = ((\textbf{l}_{i:} + \textbf{e}_i^T) - \textbf{e}_i^T) \textbf{Q}^{\mathcal{U}^{t-1}} \textbf{e}_i \\
=& (\textbf{e}_i^T - \textbf{Q}^{\mathcal{U}^{t-1}}_{i:}) \textbf{e}_i
= 1 - q^{\mathcal{U}^{t-1}}_{ii}
\end{aligned}
\end{equation*}

Thus, we rewrite equation (\ref{e:eop}) as:
\begin{equation}
\label{e:eop2}
\begin{aligned}
&\vec{\textbf{1}}^T(\textbf{Q}^{\mathcal{U}^{t-1}\cup i}-\textbf{Q}^{\mathcal{U}^{t-1}})\textbf{s}=\vec{\textbf{1}}^T\frac{\textbf{Q}^{\mathcal{U}^{t-1}} \textbf{e}_i (\textbf{e}_i^T - \textbf{Q}^{\mathcal{U}^{t-1}}_{i:})}{1-(1 - q^{\mathcal{U}^{t-1}}_{ii})}\textbf{s} \\
=& \frac{(\textbf{g}^{\mathcal{U}^{t-1}})^T \textbf{e}_i (s_i - z^{\mathcal{U}^{t-1}*}_i)}{q^{\mathcal{U}^{t-1}}_{ii}}
= \frac{g^{\mathcal{U}^{t-1}}_i}{q^{\mathcal{U}^{t-1}}_{ii}} (s_i - z^{\mathcal{U}^{t-1}*}_i) \\
\end{aligned}
\end{equation}
where we denote $\textbf{g}^{\mathcal{U}^{t-1}} = (\vec{\textbf{1}}^T \textbf{Q}^{\mathcal{U}^{t-1}})^T$ as the contribution index vector after fixing the expressed opinions of nodes in $\mathcal{U}^{t-1}$. For convenience, we denote that $\textbf{z}^{\mathcal{U}^{t-1}*}=(\textbf{z}^*|\mathcal{U}^{t-1})$.

Based on equation (\ref{e:eop2}), we integrate the calculation for each node into a simple expression:
\begin{equation}
    \textbf{b} = diag(\textbf{g}^{\mathcal{U}^{t-1}}) (diag(\textbf{Q}^{\mathcal{U}^{t-1}}))^{-1} (\textbf{s} - \textbf{z}^{\mathcal{U}^{t-1}*})
    \label{e:integration}
\end{equation}
where the $i$-th component of \textbf{b} indicates the benefit of fixing the expressed opinion of the $i$-th node and the operation $diag(*)$ expands a vector to a diagonal matrix or reverses the diagonal entries of a matrix. By using equation (\ref{e:integration}), we reduce the computational cost from $\mathcal{O}(|\mathcal{V}|^3)$ to $\mathcal{O}(|\mathcal{V}|^2)$, and this makes our SEOP method (see Algorithm 2) efficient. Equation (\ref{e:integration}) also constructs the relations between EOP, the fundamental matrix, and internal conflict \cite{chen2018quantifying}, which may help to enable the analysis of simultaneous opinion maximization and conflict reduction in future works.

\begin{algorithm}[t]
	\caption{Solving expressed opinion problem}
	\KwIn{A social trust network $\mathcal{G}=(\mathcal{V},\mathcal{E})$, internal opinions $\textbf{s}$, and budget $\mu$.}
	\KwOut{The node set $\mathcal{U}$, in which the node's expressed opinion will be fixed.}
	
	Initialize $\mathcal{U}$ to empty set;
	$t=0$;
	
	\While{$t < \mu$}
	{
		$t=t+1$;
		
		Calculate the benefits $\textbf{b}$ of fixing nodes' expressed opinions via equation (\ref{e:integration});
		
		\For {$i$ in $\mathcal{V}-\mathcal{U}^{t-1}$}{
			\uIf{$i==1$}{
				$candidate=i$;
			}
			\uElseIf{$b_i > b_{candidate}$}{
				$candidate=i$;
			}
		}
		$\mathcal{U}^t=\mathcal{U}^{t-1}\cup candidate$;
		
		Update $\textbf{Q}^{\mathcal{U}^t}$ via equation (\ref{e:updateFoundmental});
	}	
\end{algorithm}

\section{Experiments}
In this section, we describe a series of experiments that were conducted to evaluate the proposed methods. We first introduce the datasets, the experimental setup, and then the experimental results for IOP and EOP, respectively.

\subsection{Datasets}
The real-world social trust networks used in the experiments are the following: (\romannumeral1) Alpha and (\romannumeral2) OTC \cite{kumar2016edge,kumar2018rev2}. We normalized the trust values (i.e., edge weights) to the interval $[-1, 1]$. Moreover, we also tested our methods on Elec \cite{leskovec2010signed,leskovec2010predicting} and Rfa \cite{west2014exploiting}, as the relationships in these two networks are closely related to trust. The statistical details of these networks can be found in SNAP.\footnote{http://snap.stanford.edu/data/}

\subsection{Experimental Setup}
To simulate different situations, we randomly sampled values, which obey specific distributions, to initialize the internal opinion vector $\textbf{s}$. More precisely, for each network, we used five sets of internal opinions. (\romannumeral1) The internal opinions follow a uniform distribution (i.e., $\textbf{s}\sim U(-1, 1)$). (\romannumeral2) The internal opinions follow a standard normal distribution (i.e., $\textbf{s}\sim N(0, 1)$). (\romannumeral3) The absolute values of the internal opinions follow power-law distributions with $\alpha=1$ and $\alpha=2$ (i.e., $|\textbf{s}|\sim Pow(1)$ and $|\textbf{s}|\sim Pow(2)$), and each entry of $\textbf{s}$ is negated with a probability of 0.5. (\romannumeral4) The internal opinion of a node positively correlates to that node’s column connectivity (i.e., $s_i \propto \begin{matrix} \sum_{j} |a_{ji}| \end{matrix}$), and each entry of $\textbf{s}$ is negated with a probability of 0.5.

\subsection{Comparative Methods}
To the best of our knowledge, existing methods cannot solve IOP and EOP in STNs. For comparison, we modified the heuristics used in the existing studies.

For IOP, we consider four heuristics. (\romannumeral1) \textbf{Rand} \cite{li2013influence}. We randomly sort the nodes. (\romannumeral2) \textbf{Trust}. This is slightly modified from the heuristics in \cite{chen2015online}. We define the trust sum of a node as the sum of the corresponding column of the adjacency matrix. The node with large trust sum may have a strong ability to influence other nodes with their opinions. Therefore, we sort the nodes in descending order of their trust sum. (\romannumeral3) \textbf{IO}. This was inspired by \cite{musco2018minimizing}. The overall opinion may increase if we can convince the people with negative internal opinions to have positive internal opinions. Therefore, we sort the nodes in ascending order of their internal opinions. (\romannumeral4) \textbf{EO}. This was inspired by \cite{chen2018quantifying}. We sort the nodes in ascending order of their expressed opinions. After sorting the nodes, we change their internal opinions to $1$ in order.

For EOP, we consider three heuristics. (\romannumeral1) \textbf{Rand}. (\romannumeral2) \textbf{IO}. We do not consider EO in EOP, as the expressed opinion of the intervened node definitely equals its internal opinion. (\romannumeral3) \textbf{IOTS}. This is a variation of "RWR" in \cite{Gionis2013Opinion}. We consider comprehensively the internal opinion and the trust sum. More precisely, we sort the nodes in descending order of the product of the internal opinion and trust sum. After sorting, we fix the nodes' expressed opinions in order.

\begin{table}[t]
    \renewcommand\tabcolsep{3.5pt}
	\centering
	\begin{tabular}{cccccc}
		\toprule
		Dataset & \makecell[c]{Avg\\ benefit\\of SIOP} & \makecell[c]{SIOP\\vs\\Rand} & \makecell[c]{SIOP\\vs\\Trust}  & \makecell[c]{SIOP\\vs\\IO} & \makecell[c]{SIOP\\vs\\EO} \\
		\midrule
		Alpha & 273.2 & $1.66 \times$ & $2.53 \times$ & $3.41 \times$ & $2.55 \times$ \\
		OTC & 307.7 & $1.91 \times$ & $3.05 \times$ & $3.23 \times$ & $2.32 \times$ \\
		Elec & 589.9 & $3.63 \times$ & $>10 \times$ & $2.97 \times$ & $2.21 \times$ \\
		Rfa & 981.0 & $>10 \times$ & $>10 \times$ & $8.61 \times$ & $3.96 \times$ \\
		\bottomrule
	\end{tabular}
	\caption{Experimental results on IOP.}
	\label{t:IOP}
\end{table}

\begin{figure}[t]
	\centering
	\includegraphics[scale=0.37]{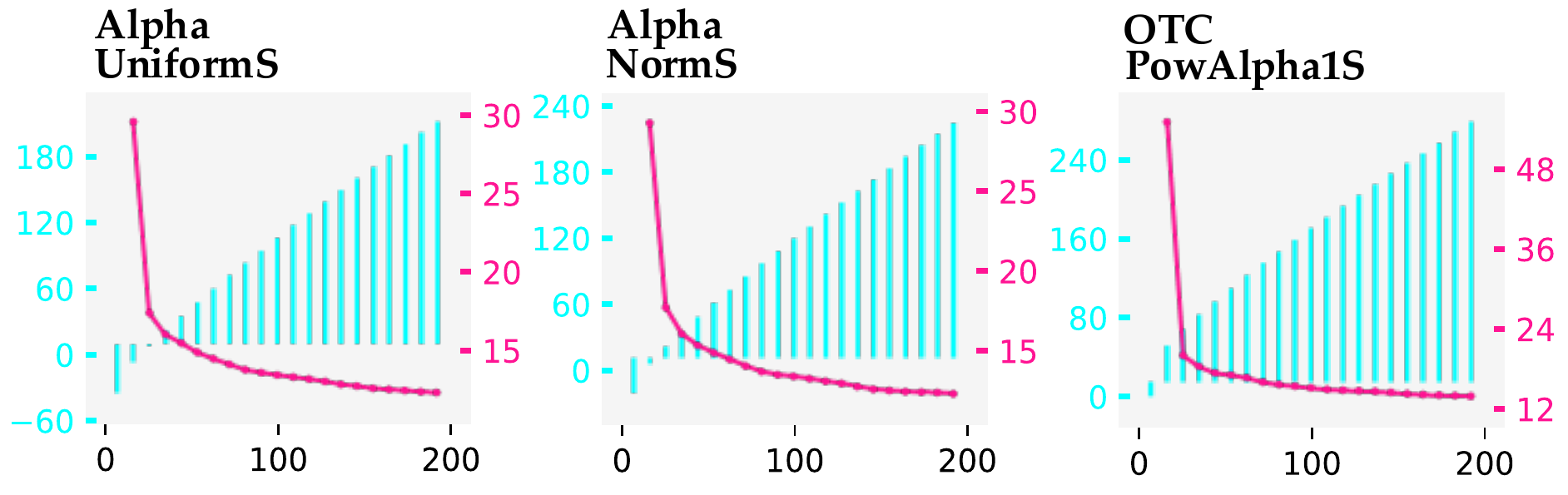}
	\caption{A part of the performance curves of SIOP. The x-axis shows the budget amount, the left y-axis shows the overall opinion, and the right y-axis shows the unit benefit.}
	\label{f:IOP-STN}
\end{figure}

\subsection{Solving Internal Opinion Problem}

We evaluated our method for solving IOP and the results are reported in Table \ref{t:IOP}. Our method SIOP has an overwhelming advantage compared to the heuristic methods, as SIOP can achieve the optimal solution. The benefit of Rand is supposed to be $\nicefrac{\mu||\textbf{g}||_{1}}{|\mathcal{V}|}$, therefore Rand can achieve good performance if the contribution index distribution is dense. The Trust method is not competitive with the other heuristic methods, as the trust sum is not the only determinant in contribution index. Actually, a node that has a high contribution index must not be influenced by many other nodes.

A part of the performance curves of SIOP, which can represent the overall performance, are shown in Figure \ref{f:IOP-STN}. We observe that different internal opinions result in very different initial overall opinions, but our method always found the best solutions in different situations. As expected, the rate of overall opinion growth decreased as the amount of modified internal opinions increased. This is because our method preferentially modifies the internal opinion of the node with the largest absolute value of contribution index in each iteration. And a greater change in the overall opinion is achieved when modifying the internal opinion of the node with a larger absolute value of contribution index. 

\begin{table}[t]
    \centering
	\begin{tabular}{ccccccc}
		\toprule
		Dataset & \makecell[c]{Avg\\benefit\\of SEOP} & \makecell[c]{SEOP\\vs\\Rand} & \makecell[c]{SEOP\\vs\\IO} & \makecell[c]{SEOP\\vs\\IOTS} \\
		\midrule
		Alpha & 192.1 & $>10 \times$ & $4.99 \times$ & $1.30 \times$ \\
		OTC & 220.6 & $>10 \times$ & $5.47 \times$ & $1.86 \times$ \\
		Elec & 65.2 & $>10 \times$ & $4.50 \times$ & $2.87 \times$ \\
		Rfa & 314.9 & $>10 \times$ & $6.92 \times$ & $3.66 \times$ \\
		\bottomrule
	\end{tabular}
	\caption{Experimental results on EOP.}
	\label{t:EOP}
\end{table}

\begin{figure}[t]
	\centering
	\includegraphics[scale=0.37]{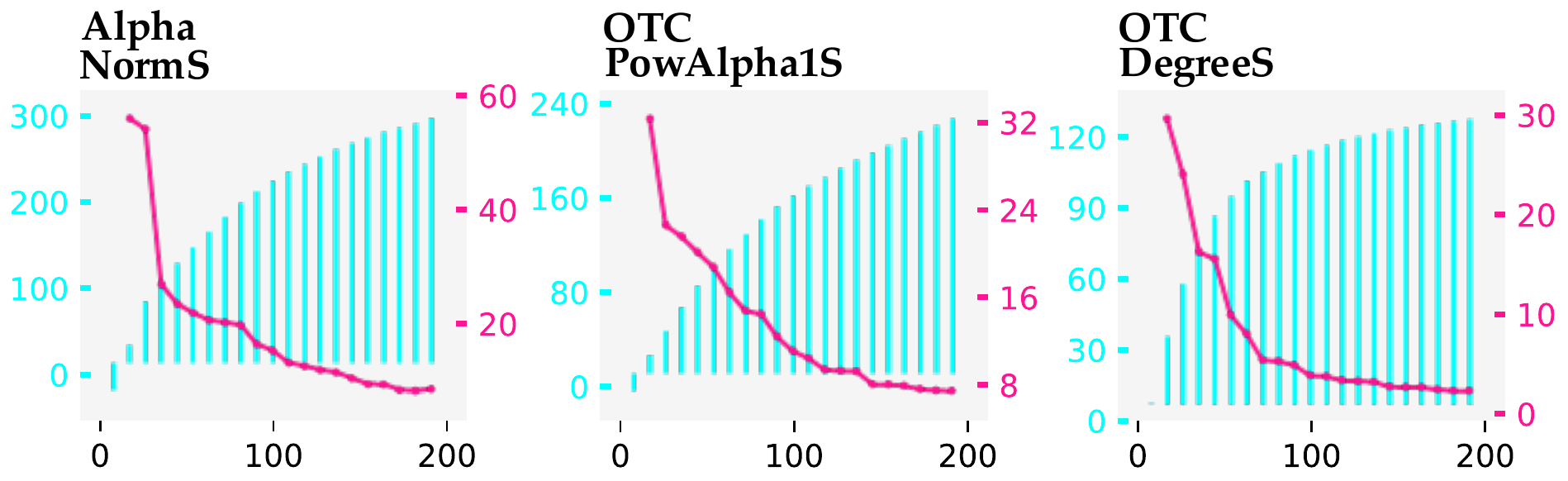}
	\caption{A part of the performance curves of SEOP.}
	\label{f:EOP-STN}
\end{figure}

\begin{table}[t!]
    \renewcommand\tabcolsep{3.5pt}
    \centering
	\begin{tabular}{ccccc}
		\toprule
		& Alpha & OTC & Elec & Rfa \\
		\midrule
		SEOP & 13.1ms & 31.1ms & 43.9ms & 93.9ms \\
		\makecell[c]{Non-optimized\\SEOP} & $4.7k \times$ & $7.3k \times$ & $>10k \times$ & $>10k \times$ \\
		\bottomrule
	\end{tabular}
	\caption{Comparison of running time (of each iteration).}
	\label{t:EOPTime}
\end{table}

\subsection{Solving Expressed Opinion Problem}
We then evaluated our method for solving EOP. From the experimental results in Table \ref{t:EOP}, we observe that our method still has an overwhelming advantage over the heuristic methods. Rand performs badly in EOP, as only a few nodes are worth intervening to get a high benefit, but it is difficult to select out these nodes by a random heuristic. Compared to IO, IOTS is a competitive heuristic method, as it considers the internal opinion and trust sum simultaneously. But there is a big gap between SEOP and IOTS, because the high-order network structure is also closely related to the benefit of fixing the expressed opinion of a node, except for the trust sum based on 1st-order structure.

A part of performance curves of SEOP, which can represent the overall performance, are shown in Figure \ref{f:EOP-STN}. We observe that in all the datasets, we can fix the expressed opinions of a few users to get very large benefits on the overall opinion. But with increasing iterations, the benefit of fixing the expressed opinion of a user declines quickly at first, and then slowly.

We also compare the running time of SEOP and SEOP without optimization (see Table \ref{t:EOPTime}).\footnote{On a server with an Intel i9-9820x CPU and 64 GB RAM.} Note that the number of iterations equals the number of nodes whose expressed opinions are fixed. In the largest network Rfa, the total running time of SEOP is about $19s$ ($\mu =200$). It is clear that the integration of vast calculations makes our method efficient.

\section{Conclusion}
In this study, we formalized two novel problems for opinion maximization in social trust networks and proposed two matrix-based methods to solve these problems. In the internal opinion problem, we observed that nodes play very different roles in the social game. We defined the contribution index, with which we could easily achieve the optimal solution. For the expressed opinion problem, we proved its hardness and designed a greedy method to achieve an acceptable solution. We were able to integrate a vast number of calculations by crafty matrix operations to create an efficient method. In real-world applications, we may have to consider numerator factors to calculate the cost of a modification, but this can be easily combined with our methods.

\section*{Acknowledgments}
This work was supported in part by the National Natural Science Foundation of China (No. 61976162, 61976161), the Key Projects of Guangdong Natural Science Foundation (No. 2018B030311003), ARC DECRA (No. DE200100964), MQRSG (No. 95109718), and Investigative Analytics Collaborative Research Project between Macquarie University and Data61 CSIRO.

\bibliographystyle{named}
\bibliography{ijcai20}

\end{document}